\documentclass[a4paper,11pt]{article}
\usepackage[utf8]{inputenc}
\usepackage{amsthm,amsmath,amssymb,amsfonts,upgreek,tikz}
\usepackage{color}

\usetikzlibrary{arrows,positioning,cd,calc}
\usepackage[pdftex, bookmarksnumbered=true]{hyperref}

\usepackage{titlesec}
\titleformat*{\section}{\Large \normalfont \bfseries}
\titleformat*{\subsection}{\large \normalfont \bfseries}
\usepackage{setspace}

\newcommand{\nc}{\newcommand}

\nc{\ny}{\nonumber}
\nc{\ra}{\rangle}
\nc{\la}{\langle}
\nc{\D}{\Delta}
\nc{\lmb}{\lambda}
\nc{\sg}{\sigma}
\nc{\ora}{\overrightarrow}
\nc{\td}{\widetilde}
\nc{\NSR}{\textsf{NSR}}
\nc{\sfF}{\textsf{F}}
\nc{\Vir}{\textsf{Vir}}
\nc{\Asl}{\widehat{\mathfrak{sl}}}
\nc{\sfV}{\textsf{V}}
\nc{\tr}{\mathrm{Tr}}
\nc{\red}{\color{red}}
\nc{\NS}{\scriptscriptstyle{\textsf{NS}}}
\nc{\R}{\scriptscriptstyle{\textsf{R}}}
\nc{\nsr}{\scriptscriptstyle{\textsf{NSR}}}
\nc{\vac}{\varnothing}
\nc{\G}{\mathsf{G}}
\nc{\BNS}{\textsf{NS}}
\nc{\BR}{\textsf{R}}
\nc{\Q}{\mathbf{Q}}
\nc{\B}{\operatorname{B}}
\nc{\ti}{\mathfrak{t}}

\numberwithin{equation}{section}

\newtheorem{prop}{Proposition}[section]
\newtheorem{lemma}{Lemma}[section]

\newtheorem{Remark}{Remark}[section]

\title{Blowup relations on $\mathbb{C}^2/\mathbb{Z}_2$ from Nakajima-Yoshioka blowup relations}
\author{Anton Shchechkin}
\date{}
\begin{document}
	\maketitle
	
	\begin{abstract}
		We obtain bilinear relations on Nekrasov partition functions,
		arising from study of tau functions of quantum $q$-Painlev\'e equations, from Nakajima-Yoshioka blowup relations by an elementary
		algebraic approach.
		
		Additionaly, using this approach, we prove certain relations on Nekrasov partition functions modified by Chern-Simons term.
		
	\end{abstract}
	
	\tableofcontents
	
	\newpage
	
	\section{Introduction}
	
	\paragraph{Background and main results.}
	This paper is motivated by studies of so-called Painlev\'e/gauge theory (or Isomonodromy/CFT)\footnote{Reason for such two names for one
	correspondence is that Painlev\'e equations arise from particular cases of isomonodromic problems on Riemann surfaces with punctures,
	and that instanton partition functions of supersymmetric gauge theories equal to certain CFT conformal blocks according to the AGT relation
	\cite{AGT09}.}
	correspondence, starting with the work \cite{GIL12}, where Painlev\'e VI tau function was written as a Fourier series
	of $SU(2)$ Nekrasov partition
	function of instantons on $\mathbb{C}^2$ with four matters and $\epsilon_1+\epsilon_2=0$
	\begin{equation}
	\tau(\sg,s|z)=\sum_{n\in\mathbb{Z}} s^n\mathcal{Z}(a+2n\epsilon_2;-\epsilon_2,\epsilon_2|z). \label{Mastertau} 
	\end{equation}
	Then plenty generalizations of this formula appeared, in particular for the tau functions of Painlev\'e V, III's equations \cite{GIL13},
	for the tau functions of $q$-Painlev\'e equations \cite{BS16q}, \cite{JNS17},\cite{MN18} as well as for isomonodromic problems,
	more sophisticated than those corresponding to Painlev\'e equations (\cite{G15}, \cite{ILT14} etc.)
	The main idea of this generalization (and of Painlev\'e/gauge theory correspondence) is that for each tau function on the Painlev\'e
	side we should relate certain instanton partition function on gauge side, such that tau function will be given by the Fourier series
	\eqref{Mastertau}.
	Particularly, for the Painlev\'e III($D_8^{(1)}$) equation it is pure gauge $SU(2)$ Nekrasov instanton partition function on $\mathbb{C}^2$,
	for $q$-Painlev\'e equations one should take 5d instanton partion functions, adding one compact dimension of radius $R=-\log q$, for
	isomonodromic problems of rank $N$ we should take $SU(N)$ gauge group etc.
	
	It turns out that Painlev\'e (differential and $q$-difference) equations and, presumably, more sophisticated isomonodromic problems
	are written as bilinear equations on these tau functions. According to \eqref{Mastertau}, such equations are equivalent to certain bilinear
	relations on Nekrasov partition functions, which have form
	\begin{equation}
	\sum_{n \in \mathbb{Z}+j/2} 
	\mathrm{D}\Big(\mathcal{Z}(a-2n\epsilon_2;-2\epsilon_2, 2\epsilon_2|z),
	\mathcal{Z}(a+2n\epsilon_2;-2\epsilon_2, 2\epsilon_2|z)\Big)=0,\quad j=0,1, \label{biltermsintro}
	\end{equation}
	where $\mathrm{D}$ is certain differential or $q$-difference operator. One of the approaches to the proof of Painlev\'e/gauge theory correspondence
	in particular cases is to find such relations on appropriate partition functions. For differential Painlev\'e equations that was done
	from the CFT side of AGT relation, using representation theory of Super Virasoro algebra (\cite{BS14},\cite{BS16b}).
	On the gauge theory side bilinear relations on Nekrasov partition functions appear from Nakajima-Yoshioka blowup relations (proved in \cite{NY05})
	\begin{equation}
	\beta^{d}_j(q_1,q_2|z)\mathcal{Z}(u;q_1,q_2|z)=\sum_{n \in \mathbb{Z}+j/2} 
	\Big(\mathcal{Z}(uq_1^{2n};q_1,q_2q_1^{-1}|q_1^{d}z)
	\mathcal{Z}(uq_2^{2n};q_1q_2^{-1},q_2|q_2^{d}z)\Big),\, q_i=e^{R\epsilon_i}\label{Z=ZZintro},
	\end{equation}
	namely, by excluding partition function in the l.h.s. from two such relations.
	However, in r.h.s. of these relations $\Omega$-background parameters differ from that in \eqref{bilconfrelintro}.
	
	Appropriate relations from the gauge theory side of AGT possibly could be obtained from the blowup relations on $\mathbb{C}^2/\mathbb{Z}_2$
	(possibly modified by 5th compact dimension). Namely,  in
	\cite{BMT11} (see also \cite{BPSS13}) 4d blowup formula was proved
	\begin{equation}
	\mathcal{Z}_{X_2}(a,\epsilon_1,\epsilon_2|z)=\sum_{n \in \mathbb{Z}} \Big(\mathcal{Z}(a+2n\epsilon_1;2\epsilon_1,-\epsilon_1+\epsilon_2|z),
	\mathcal{Z}(a+2n\epsilon_2;\epsilon_1-\epsilon_2,2\epsilon_2|z)\Big),\label{X2blowup}
	\end{equation}
	where $X_2$ is minimal resolution of $\mathbb{C}^2/\mathbb{Z}_2$.
	However, 5d modification of these blowup relations seem to be missing in the literature.
	
	First relation of such type, considered in the study of the $q$-Painlev\'e equations, is 
	\begin{equation}
	\begin{aligned}
	\sum_{n\in\mathbb{Z}}\mathcal{Z}(uq_1^{2n};q_1^2,q_2q_1^{-1}|q_1^{2}z)\mathcal{Z}(uq_1^{2n};q_1q_2^{-1},q_2^2|q_2^{2}z)
	=(1-(q_1q_2)^{1/2}z^{1/2})\sum_{n\in\mathbb{Z}}\mathcal{Z}(uq_1^{2n};q_1^2,q_2q_1^{-1}|z)\mathcal{Z}(uq_1^{2n};q_1q_2^{-1},q_2^2|z).\label{bilconfrelintro}
	\end{aligned}
	\end{equation}
	It was proposed in \cite{BS16q} (see (B.5) in loc. cit.)
	and proved in \cite{BS18} for Painlev\'e equations case $q_1q_2=1$.
	It was proved in an elementary way, using Nakajima-Yoshioka blowup relations, but it seems that used approach
	cannot be generalized for arbitrary $q_1, q_2$.
	This and analogous relations are important for study of the quantum Painlev\'e equations, namely they appear
	in Conjecture 4.2 from \cite{BGM17} and also as in \cite{BGM18} for Nekrasov partition function modified by Chern-Simons term.
	In this paper we find elementary way to obtain such relations for arbitrary $q_1, q_2$ from Nakajima-Yoshioka blowup relations.
	
	Namely, results are as follows
	\begin{itemize}
		\item We proved relations from Conjecture 4.2 from \cite{BGM17}.
		These relations are above mentioned \eqref{bilconfrelintro}, and \eqref{bilconfrel424}, \eqref{bilconfrel423}, \eqref{bilconfrel422}
		from the main text.
		\item We proved relation \eqref{bilconfrel1} on level $1$ Chern-Simons-modified Nekrasov partition functions.
		\item Using our approach, we prove certain relations on Chern-Simons modified Nekrasov partition functions,
		namely
		\begin{equation}
		\mathcal{Z}^{[2]}_{inst}(u;q_1,q_2|z)=(z;q_1,q_2)_{\infty}\mathcal{Z}^{[0]}_{inst}(u;q_1,q_2|z),\label{Z2=Z0intro}
		\end{equation}
		and \eqref{qinv1} in the main text. Here the number in the square brackets indicates the Chern-Simons level.
	\end{itemize}
	All these relations are relations on 5d $SU(2)$ pure gauge Nekrasov partition functions and we proved them for arbitrary $q_1,q_2$\footnote{
		Strictly speaking, proof is done only for the case $\epsilon_1/\epsilon_2\in\mathbb{Q}_{\leq0}$,
		because only in this case we can guarantee convergence of appropriate Nekrasov partition functions,
		see Subsection \ref{ssec:inst} for details. There is no such problem in 4d case.}.
	As we mentioned above, \eqref{bilconfrel} as well as \eqref{bilconfrel1}, were proved in \cite{BS18} for $q_1q_2=1$.
	Relation \eqref{Z2=Z0intro} was proved in \cite{BS18} for $q_1q_2=1$, $q_1q_2^2=1$, $q_1^2q_2=1$
	and relation \eqref{qinv1} was proved in \cite[Prop. 1.38]{GNY06} for $q_1^2q_2=1$. 
	It is easy to take standard limit to bilinear relations on 4d Nekrasov partition functions, we do not discuss this. 
	%Central charge of such conformal block equals
	%\begin{equation}
	%c=1+6\frac{(\epsilon_1+\epsilon_2)^2}{\epsilon_1\epsilon_2} 
	%\end{equation}
	
	\paragraph{Method.}
	The method is based on the fact, that blowing up $\mathbb{C}^2$
	twice in a way, represented on Fig.~\ref{fig:blowup} (where filed circle is point of an actual blowup) we get certain $-2$-curve
	(thick line on the scheme).~\footnote{We are grateful to Hiraku Nakajima, who suggested to use this observation to study of bilinear
		relations on Nekrasov partition functions, arising from Painlev\'e equations.}.
	Note that presence or absence of additional 5th compact dimension does not affect on blowup geometry.
	
	As mentioned above, we are interested in bilinear relations, which could be obtained from 
	blowup of $\mathbb{C}^2/\mathbb{Z}_2$. Its exceptional divisor is $-2$ curve in contrast to
	$\mathbb{C}^2$ blowup, where it is $-1$ curve.
	In terms of $\epsilon_1,\epsilon_2$ this is represented by the value $I=\epsilon_1^{(1)}/\epsilon_2^{(1)}+\epsilon_2^{(2)}/\epsilon_1^{(2)}$,
	where $\epsilon_{1,2}^{(\eta)}, \eta=1,2$ are $\Omega$-background parameters for the first and the second partition functions in blowup equation
	respectively. For Nakajima-Yoshioka blowup relation \eqref{Z=ZZintro} this equals to $-1$, and for $X_2$ blowup relation \eqref{X2blowup}
	it equals to $-2$. Therefore we will call bilinear relations on Nekrasov partition functions with $I=-2$ as $-2$ bilinear relations.
	
	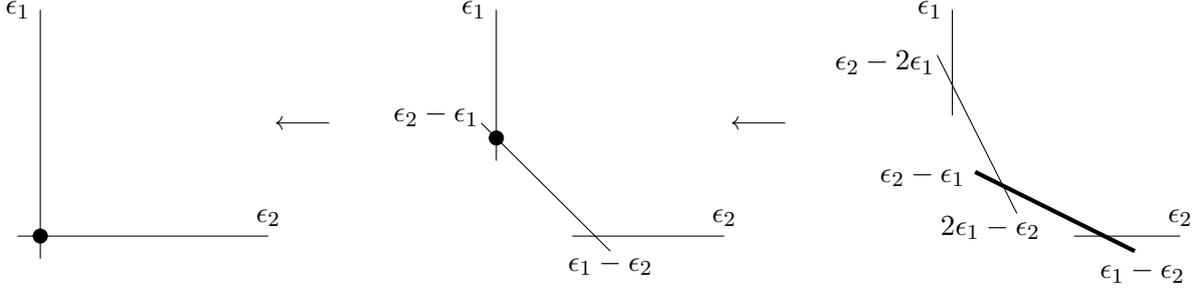
\begin{figure}[h]
		\begin{center}
			\begin{tikzpicture}
			\draw (1.7,1) -- (5,1);
			\draw (2,0.7) -- (2,4);
			\draw (2,4) node [anchor=east] {$\epsilon_1$};
			\draw (5,1) node [anchor=south] {$\epsilon_2$};
			\fill (2,1) circle (0.1cm);
			
			\draw [<-] (5.1,2.5) -- (5.8,2.5);

			\draw (9,1) -- (11,1);
			\draw (8,2) -- (8,4);
			\draw (7.8,2.5) -- (9.5,0.8);
			\draw (8,4) node [anchor=east] {$\epsilon_1$};
			\draw (11,1) node [anchor=south] {$\epsilon_2$};
			\draw (7.9,2.6) node [anchor=east] {$\epsilon_2-\epsilon_1$};
			\draw (9.5,0.9) node [anchor=north] {$\epsilon_1-\epsilon_2$};
			
			\fill (8,2.3) circle (0.1cm);
			
			\draw[<-] (11.1,2.5) -- (11.8,2.5);
			
			\draw (15.6,1) -- (17,1);
			\draw[ultra thick] (16.4,0.8) -- (14.3,1.85);
			\draw (14,2.6) -- (14,4);
			\draw (13.8,3.4) -- (14.85,1.3);
			
			\draw (14,4) node [anchor=east] {$\epsilon_1$};
			\draw (17,1) node [anchor=south] {$\epsilon_2$};
			\draw (13.9, 3.3) node [anchor=east] {$\epsilon_2-2\epsilon_1$};
			\draw (14.3, 1.8) node [anchor=east] {$\epsilon_2-\epsilon_1$};
			\draw (14.5, 1.4) node [anchor=north] {$2\epsilon_1-\epsilon_2$};
			\draw (16.5, 0.8) node [anchor=north] {$\epsilon_1-\epsilon_2$};
			
			\end{tikzpicture}
			\caption{\leftskip=1in \rightskip=1in Blowup scheme of the approach}
			\label{fig:blowup}
		\end{center}
	\end{figure}
	
	Using this observation, we try to represent it at the level of instanton partition functions.
	Practically our approach breaks down into several steps:
	\begin{enumerate}
		\item We take conjectured $-2$ bilinear relation $\widehat{\mathcal{Z}}(u;q_1,q_2)=0$
		and convolute it with another Nekrasov partition function (here we omit dependence on $z$ for simplicity)
		\begin{equation}
		\sum_{m\in\mathbb{Z}}\mathcal{Z}(u(q_1q_2)^{2m};q_1q_2,q_1^{-2}) \widehat{\mathcal{Z}}(uq_2^{2m};q_1,q_2)
		\label{convolintro}.
		\end{equation}
		
		\item Substituting $\widehat{\mathcal{Z}}(q_1,q_2)=\sum_{terms} \sum_{n\in\mathbb{Z}} \mathcal{Z}(q_1^2,q_2 q_1^{-1}) 
		\mathcal{Z}(q_1 q_2^{-1},q_2^2) $ and using Nakajima-Yoshioka blowup relations \eqref{Z=ZZintro} twice, we obtain linear combination of
		Nekrasov partition functions
		\begin{multline}
		\!\!\!\!\!\!\!\!\!\!\!\!\!\!\!\!\!\!\sum_{terms} \sum_{m,n\in\mathbb{Z}}\mathcal{Z}(q_1q_2,q_1^{-2}) \mathcal{Z}(q_1^2,q_2 q_1^{-1}) \mathcal{Z}(q_1 q_2^{-1},q_2^2)=
		\sum_{terms} \sum_{m\in\mathbb{Z}}\mathcal{Z}(q_1q_2,q_2 q_1^{-1}) \mathcal{Z}(q_1 q_2^{-1},q_2^2)
		=\sum_{terms}  \mathcal{Z}(q_1 q_2,q_2^2),
		\end{multline}
		where we also omitted dependence on $u$, which is shifted appropriately. First sum represents several terms in the initial bilinear relation.
		\item Suppose that this linear combination is zero. Finally we prove, that initial $-2$ bilinear
		relation was also zero. It is done step by step, proving that all $z$-powers of the initial relations equal zero.
	\end{enumerate}
	
	Note that connection between $-1$ and $-2$ blowups is used in the theory of Donaldson invariants \cite{FS94}, \cite{B94}.

	\paragraph{Content.}
	In Section \ref{sec:Nekrasov} we recall necessary facts about Nekrasov partition functions and Nakajima-Yoshioka
	blowup relations. There we also prove convergence of 5d pure gauge $SU(2)$ Nekrasov partition function for 
	$\epsilon_1/\epsilon_2\in\mathbb{Q}_{<0}$.
	
	Section \ref{sec:-2} describes our approach to obtain $-2$ bilinear relations on Nekrasov partition functions.
	We start from general scheme in Subsection \ref{ssec:gs}, then we obtain $-2$ bilinear relations
	in even (Subsection \ref{ssec:even}) and odd sector (Subsection \ref{ssec:odd}). We finish by 
	generalization of our approach for Chern-Simons-modified Nekrasov partition functions.
	
	In Section \ref{sec:appl} we obtain higher order Nakajima-Yoshioka blowup relations as a byproduct of our approach
	and prove certain symmetry relations on Chern-Simons-modified Nekrasov partition functions.
	
	We conclude with several directions for further study.

	\paragraph{Acknowledgements.}
	We thank Hiraku Nakajima for telling the idea, that inspired writing this paper, and pointing
	us to references \cite{FS94}, \cite{B94}, Mikhail Bershtein for interest to our work and
	stimulating discussions, Roman Gonin for discussion of Prop. \ref{prop:suff}.
	
	We are grateful to Pavlo Gavrylenko and Mykola Semenyakin for a careful reading
	of the Introduction.
	
	This work is partially supported by HSE University Basic Research Program and funded
        (partially) by the Russian Academic Excellence Project ’5-100’.
	Classification of $-2$ bilinear relations was supported by the Russian Science Foundation
	(project 16-11-10316).
	
	\section{Nekrasov functions and Nakajima-Yoshioka blowup relations}
	\label{sec:Nekrasov}
	\subsection{5d Nekrasov partition function}
	
	We start from reviewing Nekrasov partition functions $\mathcal{Z}$ of pure SUSY $SU(2)$ gauge theory on $\mathbb{C}^2$ extended 
	by the 5th compact dimension and discuss their components. 
	
	Full Nekrasov partition function $\mathcal{Z}$ splits into three factors (we follow conventions of \cite{NY03L}, \cite{NY05})
	\begin{equation}
	\mathcal{Z}=\mathcal{Z}_{cl}\mathcal{Z}_{1-loop}\mathcal{Z}_{inst}.\label{Zstr}
	\end{equation}
	In loc. cit. $\mathcal{Z}_{cl}$ and $\mathcal{Z}_{1-loop}$ appear from the so-called "perturbative" part. 
	Nekrasov function depends on parameters of the $\Omega$-background $\epsilon_1, \epsilon_2$, vacuum expectation values $a_1, a_2$
	with condition $a_1+a_2=0$ (we denote $a=a_1-a_2$) and also on the radius $R$ of the 5th compact dimension.
	In 5d case it is convenient to use multiplicative parameters, connected with above by
	\begin{equation}
	u_i=e^{Ra_i}, \quad q_i=e^{R\epsilon_i},\quad i=1,2
	\end{equation}
	with condition $u_1u_2=1$
	(we denote $u=u_1/u_2$). 
	To obtain pure SUSY $SU(2)$ gauge theory on $\mathbb{C}^2$ one should tend $R\rightarrow0$, we will not discuss such limit in this paper. 
	Sometimes we also want to modify our pure SUSY $SU(2)$ gauge theory by additional Chern-Simons theory of level $l\in\mathbb{Z}$.
	We denote such instanton partition function (and related objects) by superscript $[l]$. 
	
	In this paper we will consider only region $\epsilon_1<0<\epsilon_2$, which corresponds to central charge $c\leq1$ under the AGT correspondence,
	due to several reasons. This region contains cases $\epsilon_1+\epsilon_2=0$ and $2\epsilon_1+\epsilon_2=0$ which are interesting in context of applications to Painlev\'e
	equations, as we explained in the Introduction. This region is also closed under Nakajima-Yoshioka blowup relations (see \eqref{eq:Z=ZZ} below),
	which we discuss at the end of this Section.
	
	\paragraph{Instanton part of Nekrasov partition functions.}
	\label{ssec:inst}
	Instanton part of 5d Nekrasov function, modified by Chern-Simons theory of level $l$ is given by Nekrasov formula
	\begin{equation}
	\mathcal{Z}^{[l]}_{inst}(u;q_1,q_2|z)=\sum_{\lmb^{(1)},\lmb^{(2)}}\frac{\prod_{i=1}^2 (q_1q_2)^{-\frac{l}2 |\lmb^{(i)}|}
		\mathsf{T}^l_{\lmb^{(i)}}(u_i;q_1,q_2)}
	{\prod_{i,j=1}^2\mathsf{N}_{\lmb^{(i)},\lmb^{(j)}}(u_i/u_j;q_1,q_2)}(q_1^{-1}q_2^{-1}z)^{|\lmb^{(1)}|+|\lmb^{(2)}|},
	\label{Zinst5}
	\end{equation}
	written in terms of combinatorial block
	\begin{equation}
	\mathsf{N}_{\lmb,\mu}(u;q_1,q_2)=\prod_{s\in\mathbb{\lmb}}\left(1-uq_2^{-a_{\mu}(s)-1} q_1^{l_{\lmb}(s)}\right)
	\prod_{s\in\mathbb{\mu}}\left(1-uq_2^{a_{\lmb}(s)} q_1^{-l_{\mu}(s)-1}\right) \label{N5}
	\end{equation}
	and Chern-Simons term 
	\begin{equation}
	\mathsf{T}_{\lmb}(u;q_1,q_2)=\prod_{(i,j)\in\lmb} u^{-1}q_1^{1-i}q_2^{1-j}. \label{CSterm}
	\end{equation}
	Here $\lmb^{(1)}, \lmb^{(2)}$ are partitions, $|\lmb|=\sum\lmb_j$ and $a_{\lmb}(s), l_{\lmb}(s)$ denote lengths of arms and legs
	for the box $s$ in the Young diagram corresponding to the partition $\lmb$. 
	
	The function $\mathcal{Z}^{[l]}_{inst}(u;q_1,q_2|z)$ satisfies elementary symmetry properties:
	\begin{equation}
	\mathcal{Z}^{[l]}_{inst}(u;q_1,q_2|z)=\mathcal{Z}^{[l]}_{inst}(u;q_2,q_1|z)=\mathcal{Z}^{[l]}_{inst}(u^{-1};q_1,q_2|z). \label{qZsymm} 
	\end{equation}
	For $l=0$ there is also elementary symmetry 
	\begin{equation}
	\mathcal{Z}_{inst}(u;q_1,q_2|z)=\mathcal{Z}_{inst}(u;q_1^{-1},q_2^{-1}|z). \label{qinv}
	\end{equation}
	Its proof is based on term by term coincidence of the power series, that's why we have immediately 
	\begin{equation}
	\mathcal{Z}^{[-l]}_{inst}(u;q_1,q_2|z)=\mathcal{Z}^{[l]}_{inst}(u;q_1^{-1},q_2^{-1}|z) \label{CSlinv}
	\end{equation}
	In the case $q_1q_2=1$ the symmetry $q_1, q_2\mapsto q^{-1}_1, q^{-1}_2$
	is equivalent to the symmetry $q_1\leftrightarrow q_2$ for arbitrary $l$.
	
	For general $q_1, q_2$, the situation with $q_1, q_2\mapsto q^{-1}_1, q^{-1}_2$ symmetry is more subtle.
	For $l\neq0$ term by term comparison does not work.
	For $l=1$, however, one has
	\begin{equation}
	\mathcal{Z}^{[1]}_{inst}(u;q_1,q_2|z)=\mathcal{Z}^{[1]}_{inst}(u;q_1^{-1},q_2^{-1}|z).\label{qinv1}
	\end{equation}
	The proof for the case $q_1=q^{-1},q_2=q^2$ case is given in \cite[Prop. 1.38]{GNY06}. 
	We will prove this equality in Subsection \ref{ssec:bruteforce} for arbitrary $q_1,q_2$.
	
	For $l=2$ it turns out that
	\begin{equation}
	\mathcal{Z}^{[2]}_{inst}(u;q_1,q_2|z)=(z;q_1,q_2)_{\infty}\mathcal{Z}^{[0]}_{inst}(u;q_1,q_2|z), \label{Z2=Z0}
	\end{equation}
	so that 
	\begin{equation}
	\mathcal{Z}^{[2]}_{inst}(u;q_1^{-1},q_2^{-1}|z)=\frac{1-z}{(z;q_1)_{\infty}(z;q_2)_{\infty}}\mathcal{Z}^{[2]}_{inst}(u;q_1,q_2|za). 
	\end{equation}
	We will discuss and prove equality \eqref{Z2=Z0} in Subsection \ref{ssec:bruteforce}.
	
	\paragraph{Convergence of Nekrasov partition functions.}
	Let us consider the convergence of the series \eqref{Zinst5}. We proved that
	\begin{prop}\label{prop:convF}
		Let $q_1=q^{-m}, q_2=q^{n}$, $m,n\in\mathbb{Z}_{\geq1}$,  
		$|q|\neq 1$ and $u \neq q^k$, $k \in \mathbb{Z}$. Then series \eqref{Zinst5} for $l=0$ 
		converges uniformly and absolutely on every bounded subset
		of $\mathbb{C}$.
	\end{prop}
	The proof of this Proposition is similar to the one in \cite[Prop 1 (i)]{ItsLTy14}
	and generalize Proposition \cite[Prop. 3.1.]{BS16q} to the case ${\epsilon_1}/{\epsilon_2}\in\mathbb{Q}_{<0}$.
	\begin{proof}
		There exist constants $L_1,L_2 \in \mathbb{R}_{>0}$, such that
		\begin{equation}\left|\frac{q^{k/2}-q^{-k/2}}{q^{1/2}-q^{-1/2}}\right|>|k| L_1^{1/2},\;\; \forall k \in \mathbb{Z}_{\neq 0},\qquad 
		\left|\frac{u^{1/2}q^{k/2}-u^{-1/2}q^{-k/2}}{q^{1/2}-q^{-1/2}}\right|>L_2^{1/2}, \;\; \forall k \in \mathbb{Z}.\end{equation}
		Then we can bound
		$\prod_{i,j=1}^2 N_{\lambda_i,\lambda_j}(u_i/u_j;q^{-m},q^n)$ as 
		\begin{multline}
		\Bigl| N_{\lambda_1,\lambda_1}(1;q^{-m},q^n) N_{\lambda_2,\lambda_2}(1;q^{-m},q^n) \Bigr|=\\
		\prod_{s \in \lambda_1}|q|^{\frac{m-n}2}|q^{\frac12 (n(a_{\lambda_1}(s)+1)+ml_{\lambda_1}(s))}-q^{-\frac12 (n(a_{\lambda_1}(s)+1)+ml_{\lambda_1}(s))}|
		|q^{\frac12 (na_{\lambda_1}(s)+m(l_{\lambda_1}(s)+1))}-q^{-\frac12 (na_{\lambda_1}(s)+m(l_{\lambda_1}(s)+1))}|\\
		\cdot \prod_{s \in\lambda_2}(\lambda_1 \leftrightarrow \lambda_2) >
		\left(\prod_{s \in \lambda_1} |h_{\lambda_1}|^2\right)\left(\prod_{s \in \lambda_2} |h_{\lambda_2}|^2\right) 
		\left|q^{\frac{m-n}2}\operatorname{min}^2(m,n)L_1(q^{1/2}-q^{-1/2})^2\right|^{|\lambda_1|+|\lambda_2|}>\\
		> \frac{|\lambda_1|!^2|\lambda_2|!^2}{(\dim\lambda_1\dim\lambda_2)^2}\left|q^{\frac{m-n}2}\operatorname{min}^2(m,n)L_1(q^{1/2}-q^{-1/2})^2\right|^{|\lambda_1|+|\lambda_2|},
		\end{multline}
		\begin{multline}
		\Bigl|N_{\lambda_1,\lambda_2}(u;q^{-m},q^n) N_{\lambda_2,\lambda_1}(u^{-1};q^{-m},q^n)\Bigr| =\\=
		\prod_{s \in \lambda_2}|q|^{\frac{m-n}2}|u^{\frac12}q^{\frac12(n(a_{\lambda_2}(s)+1)+m\l_{\lambda_1}(s))}-u^{-\frac{1}2}q^{-\frac12(n(a_{\lambda_2}(s)+1)+ml_{\lambda_1}(s))}|\\
		\cdot|u^{\frac12}q^{\frac12(na_{\lambda_2}(s)+m(l_{\lambda_1}(s)+1))}-u^{-\frac{1}2}q^{-\frac12(na_{\lambda_2}(s)+m(l_{\lambda_1}(s)+1))}|\\ 
		\cdot \prod_{s \in\lambda_1}(\lambda_1 \leftrightarrow \lambda_2, u\leftrightarrow u^{-1}) >  
		\left|q^{\frac{m-n}2}L_2(q^{1/2}-q^{-1/2})^2\right|^{|\lambda_1|+|\lambda_2|},
		\end{multline}
		where we used hook length formula for $\dim \lambda$.  
		Since $\sum_{|\lambda|=n} (\dim\lambda)^2=n!$ we have $\mathcal{Z}_{inst}(u;q^{-1},q|z)<\exp\left|\dfrac{2z}{q^{m-n}\operatorname{min}^2(m,n)L_1 L_2(q^{1/2}-q^{-1/2})^4}\right|$.
	\end{proof}
	For $\epsilon_1/\epsilon_2\notin\mathbb{Q}$ arguments of above proof do not work.
	In this case poles in $u$ of sum \eqref{Zinst5} are dense and it seems that the series diverges.
	However, we do not have any proof.
	For $l\neq0$ we also don't know any proof, however numerical experiments suggest that it diverges 
	when $|l|>2$. According to symmetry \eqref{CSlinv} below we will restrict ourselves to the levels $l=0,1,2$.
	Such restriction is also natural from the cluster point of view \cite{BGM18}.
	We will consider case $l=2$ only to discuss relation \eqref{Z2=Z0}.
	
	Note also that results on convergence of Nekrasov instanton partition functions in other sector
	(namely $\epsilon_1,\epsilon_2>0$ and also complex conjugated non-imaginary $\epsilon_1,\epsilon_2$) 
	were obtained in paper\cite{FML17}.
	
	\paragraph{Classical and 1-loop part of Nekrasov partition functions.}
	Classical and 1-loop parts of 5d Nekrasov function are given by
	\begin{eqnarray}
	\mathcal{Z}_{cl}(u;q_1,q_2|z)&=(q_1^{-1}q_2^{-1}z)^{-\frac{\log^2 u}{4\log q_1\log q_2}},\label{Zcl5}\\
	\mathcal{Z}_{1-loop}(u;q_1,q_2)&=(u;q_1,q_2)_{\infty}(u^{-1};q_1,q_2)_{\infty},\label{Z1loop5}
	\end{eqnarray}
	where $q$-Pochhammer symbol defined by 
	\begin{equation}
	(z;q_1,\ldots q_N)_{\infty}=\prod_{i_1,\ldots i_N=0}^{\infty}\left(1-z\prod_{k=1}^Nq_k^{i_k}\right)
	\label{Pochhammer_def}
	\end{equation}
	satisfy $q$-shift relations
	\begin{equation}
	(z;q_1,\ldots q_N)_{\infty}/(zq_1;q_1,\ldots q_N)_{\infty}=(z;q_2,\ldots q_N)_{\infty},
	\quad (z;q)_{\infty}/(zq;q)_{\infty}=1-z. \label{qshift} 
	\end{equation}
	These parts do not depend on $l$ (however, $\mathcal{Z}_{cl}$ becomes depend on $l$ in case of $SU(r)$, $r>2$ gauge group).
	One can see that symmetries $q_1\leftrightarrow q_2$, $u\mapsto u^{-1}$ are also satisfied by classical and 1-loop parts of the full
	Nekrasov function $\mathcal{Z}$. However, for the symmetry $q_1, q_2\mapsto q_1^{-1}, q_2^{-1}$ we have
	\begin{equation}
	\mathcal{Z}_{cl}(u;q_1^{-1},q_2^{-1}|z)=(q_1q_2)^{-\frac{\log^2 u}{2\log q_1\log q_2}}\mathcal{Z}_{cl}(u;q_1,q_2|z), \label{qclassym}
	\end{equation} 
	and
	\begin{equation}
	\mathcal{Z}_{1-loop}(u;q_1^{-1},q_2^{-1})=(uq_1;q_1)^{-1}_{\infty} (u^{-1};q_1)^{-1}_{\infty}  (u;q_2)^{-1}_{\infty} (q_2u^{-1};q_2)^{-1}_{\infty} \mathcal{Z}_{1-loop}(u;q_1,q_2) \label{qloopassym}
	\end{equation} 
	(where we used properties $(z;q_1^{-1},q_2,\ldots q_N)_{\infty}=(zq_1;q_1,\ldots q_N)^{-1}_{\infty}$ and \eqref{qshift} successively),
	so symmetry is broken for all cases except $q_1q_2=1$.
	
	\subsection{Nakajima-Yoshioka blowup relations}
	\label{ssec:NYblowup}
	
	Functions $\mathcal{Z}^{[l]}(u;q_1,q_2|z)$ are known to satisfy Nakajima-Yoshioka blowup relations \cite{NY05}, \cite{GNY06}
	\begin{equation}\label{eq:Z=ZZ}
	\beta^{d}_j(q_1,q_2|z)\mathcal{Z}^{[l]}(u;q_1,q_2|z)=\sum_{n \in \mathbb{Z}+j/2} 
	\Big(\mathcal{Z}^{[l]}(uq_1^{2n};q_1,q_2q_1^{-1}|q_1^{d+\frac12l(j-1)}z)
	\mathcal{Z}^{[l]}(uq_2^{2n};q_1q_2^{-1},q_2|q_2^{d+\frac12l(j-1)}z)\Big), 
	\end{equation}
	for $j=0,l\in\mathbb{Z}/2\mathbb{Z}$, $d=-1,0,1$. Such coefficcients $\beta^{d}_j$ turned out to be independent from $l$, they are given in the table
	\medskip
	
	\begin{tabular}{|l|l|l|l|}
		\hline
		$\beta^d_j$ & $d=-1$ & $d=0$ & $d=1$ \\
		\hline
		$j=0$ & 1 & 1 & 1 \\
		\hline
		$j=1$ & $(q_1^{-1}q_2^{-1}z)^{1/4}$ & 0 &  $-(q_1q_2z)^{1/4}$ \\
		\hline
	\end{tabular}
	\medskip
	
	These results were proved in Theorem 2.4 in \cite{NY05} for $l=0$ and in Theorem 2.11 in \cite{NY09} for the case $l=1,2$,
	$j=0$ and case $l=1,2$, $j=1$, $d=0$.
	We have not found cases $l=1$, $j=1$, $d=\pm1$ in the literature but they follow from the results of
	\cite{NY09} (see footnote $6$ in \cite{BS18}) %{\color{red} l=2?}
	One could ask about higher order Nakajima-Yoshioka blowup relations (with $|d|>1$). We will show, how to obtain
	these relations (especially $d=\pm2$) in Subsection \ref{ssec:hoNY}. 
	
	Note that, in fact, Nakajima-Yoshioka blowup relations are relations on $\mathcal{Z}_{inst}^{[l]}$,
	and $\mathcal{Z}_{cl}$ and $\mathcal{Z}_{1-loop}$ give $(\textbf{l}^d_n)^{-1}z^{n^2}$, where $z$-independent coefficcient $\textbf{l}^d_n$
	is called blowup factor.  
	
	\begin{Remark}\label{rem:binv}
		According to symmetries \eqref{CSlinv}, \eqref{qclassym}, \eqref{qloopassym}
		coefficients $\beta_j^{d,[l]}$ for arbitrary $d$ satisfy
		\begin{equation}
		\beta_j^{d,[l]}(q_1^{-1},q_2^{-1}|z)=(-1)^j \beta_j^{-d,[-l]}(q_1,q_2|z),\label{betasymm}
		\end{equation}
		where we restored $\beta_j^{d,[l]}$ dependence on $l$ for arbitrary $d$.  
	\end{Remark}

	\section{Blowup relations on $\mathbb{C}^2/\mathbb{Z}_2$ from blowup relations on $\mathbb{C}^2$}
	\label{sec:-2}
	\subsection{General scheme}
	\label{ssec:gs}
	There are also blowup relations on Nekrasov partition functions on $\mathbb{C}^2/\mathbb{Z}_2$ which have form
	\begin{equation}\label{eq:Z2=ZZ}
	\mathcal{Z}^{[l]}_{X_2}(u;q_1,q_2|z)=\sum_{n \in \mathbb{Z}+j/2}
	\mathrm{D}\Big(\mathcal{Z}^{[l]}(uq_1^{2n};q_1^2,q_2q_1^{-1}|z),\mathcal{Z}^{[l]}(uq_2^{2n};q_1q_2^{-1},q_2^2|z)\Big), \, j=0,1
	\end{equation}
	where $\mathcal{Z}^{[l]}_{X_2}$ is certain instanton partition function on $X_2$, which is minimal resolution 
	of $\mathbb{C}^2/\mathbb{Z}_2$. Explicit type of this partition function depends on bilinear $q$-difference in $z$ operator $\mathrm{D}$ in r.h.s. 
	
	As already mentioned in Introduction, such relations can be used to prove formula \eqref{Mastertau} for tau
	functions of $q$-difference Painlev\'e equations.
	However, to do this one needs only bilinear relations on Nekrasov partition functions, which can be obtained
	from the above relations by exluding $\mathcal{Z}^{[l]}_{X_2}$.
	These relations are given by sum of following terms
	\begin{equation}
	\widehat{\mathcal{Z}}^{[l]}_d(u;q_1,q_2|z)=\sum_{n\in\mathbb{Z}+j/2}\epsilon^n
	\mathcal{Z}^{[l]}(uq_1^{2n};q_1^2,q_2q_1^{-1}|q_1^{d}z)
	\mathcal{Z}^{[l]}(uq_2^{2n},q_1q_2^{-1},q_2^2|q_2^{d}z), \quad d\in\mathbb{Z},\label{bilterms}
	\end{equation}
	with coefficcients, independent from $u$.
	We will see below that sign $\epsilon=\pm1$ is essential only in case $j=1$. 
	
	We call such bilinear relations "$-2$ bilinear relations", emphasize that for these relations
	$ \epsilon_1^{(1)}/\epsilon_2^{(1)}+\epsilon_2^{(2)}/\epsilon_1^{(2)}=-2$ in contrast to
	Nakajima-Yoshioka blowup relations, where such l.h.s. equals $-1$. 
	Our aim is to find approach to derive $-2$ bilinear relations from Nakajima-Yoshioka blowup relations \eqref{eq:Z=ZZ}.
	In this Subsection we will illustrate general scheme on the case, when $l=0$. Generalization
	for $l=1$ will be given in Subsection \ref{ssec:CS}. 
	
	Let us make a certain convolution of $\mathcal{Z}(u;q_1q_2,q_1^{-2}|z)$
	with \eqref{bilterms}
	\begin{equation}
	\sum_{m\in\mathbb{Z}+j/2}\epsilon^m \mathcal{Z}(u(q_1q_2)^{2m};q_1q_2,q_1^{-2}|(q_1q_2)^{d_1+d_2}z) \widehat{\mathcal{Z}}_d(uq_2^{2m};q_1,q_2|q_2^{d_1+d_2}z)
	\label{convol}
	\end{equation}
	or, explicitly writing $-2$ bilinear term
	\begin{equation}
	\begin{aligned}
	\sum_{m,n\in\mathbb{Z}+j/2}\epsilon^{m+n}\mathcal{Z}(uq_1^{2m}q_2^{2m};q_1q_2,q_1^{-2}|(q_1q_2)^{d_1+d_2}z)\\
	\times\mathcal{Z}(uq_1^{2n}q_2^{2m};q_1^2,q_2q_1^{-1}|q_1^{d_2-d_1}q_2^{d_1+d_2}z)
	\mathcal{Z}(uq_2^{2(m+n)},q_1q_2^{-1},q_2^2|q_2^{d_2-d_1}q_2^{d_1+d_2}z), \label{convolexpl}
	\end{aligned}
	\end{equation}
	where we introduced $d_1$ and $d_2$, such that $d=d_2-d_1$.
	Let us make substitution $m=m'+n', n=m'-n'$ in this expression 
	\begin{equation}
	\begin{aligned}
	\epsilon^j \left(\sum_{m'\in\mathbb{Z}+j/2,n'\in\mathbb{Z}}+\epsilon\sum_{m'\in\mathbb{Z}+j/2+\frac12,n'\in\mathbb{Z}+\frac12}\right)
	\mathcal{Z}(u(q_1q_2)^{2(m'+n')};q_1q_2,q_1^{-2}|(q_1q_2)^{d_2+d_1}z)\\
	\mathcal{Z}(u(q_1q_2)^{2m'} (q_1^{-1}q_2)^{2n'};q_1^2,q_2q_1^{-1}|q_1^{d_2-d_1}
	q_2^{d_2+d_1}z)
	\times\mathcal{Z}(uq_2^{4m'},q_1q_2^{-1},q_2^2|q_2^{2d_2}z)
	\end{aligned}
	\end{equation}
	Using Nakajima-Yoshioka blowup relations \eqref{eq:Z=ZZ} for the first pair of Nekrasov partition functions (summing up over $n'$), we obtain
	\begin{equation}
	\begin{aligned}
	\sum_{2m'\in\mathbb{Z}}\epsilon^{2m'}\beta^{d_1}_{2m'+j \mod 2}(q_1q_2,q_2q_1^{-1}|(q_1q_2)^{d_2}z)
	\mathcal{Z}(u(q_1q_2)^{2m'};q_1q_2,q_2q_1^{-1}|(q_1q_2)^{d_2} z)
	\mathcal{Z}(uq_2^{4m'} ,q_1q_2^{-1},q_2^2|q_2^{2d_2}z),
	\end{aligned}
	\end{equation}
	which could be summed up again, using Nakajima-Yoshioka blowup relations \eqref{eq:Z=ZZ}, to
	\begin{equation}
	\left(\sum_{i=0,1}\epsilon^i\beta^{d_1}_{i+j \mod 2}(q_1q_2,q_2q_1^{-1}|(q_1q_2)^{d_2}z) \beta^{d_2}_{i}(q_1q_2,q_2^2|z)\right)\mathcal{Z}(u;q_1q_2,q_2^2|z)\label{coeff}
	\end{equation}
	
	Let us take sum of terms \eqref{bilterms} and make convolution \eqref{convol} with the whole sum.
	Assume that after such convolution we obtain $\mathcal{Z}(u;q_1q_2,q_2^2|z)$ with zero coefficient.
	Note that to make such convolution, uniform for all terms, we should take terms in the sum with the same shift $d_1+d_2$.
	
	Below we prove that the initial sum is also zero in case $j=0$. Case $j=1$ is more subtle: we prove that initial sum vanishes, when
	there are two different convolutions of the initial sum. In principle, we have such possibility, because we could take different pairs
	$d_1$ and $d_2$, s.t. $d_2-d_1=d$. Below for simplicity we denote $d_{12}=d_1+d_2$.
	
	\begin{prop}\label{prop:suff}
		Consider $-2$ bilinear relation on Nekrasov partition functions which is sum of terms \eqref{bilterms}.
		
		(i) Case $j=0$. If there is a vanishing convolution, given by \eqref{convol}, then initial $-2$ bilinear relation is satisfied. 
		
		(ii) Case $j=1$. If there is at least two vanishing convolutions, s.t. $d_{12}\neq d_{12}'$, then initial $-2$ bilinear relation is satisfied. 
	\end{prop}
	\begin{proof}
		In terms of $\mathcal{Z}_{inst}$, summands \eqref{bilterms} have the form
		\begin{equation}
		\begin{aligned}
		(q_1^{-1}q_2^{-1}z)^{-\frac{\log^2u}{8\log q_1 \log q_2}}\sum_{n\in\mathbb{Z}+j/2}
		\frac{\epsilon^n u^{\frac{nd}2} (q_1q_2)^{\frac{n^2d}2} (q_1^{-1}q_2^{-1}z)^{\frac{n^2}2}}{\mathfrak{l}^d_n}\mathcal{Z}_{inst}(uq_1^{2n};q_1^2,q_2q_1^{-1}|q_1^{d}z)
		\mathcal{Z}_{inst}(uq_2^{2n},q_1q_2^{-1},q_2^2|q_2^{d}z),
		\end{aligned}
		\end{equation}
		where $\mathfrak{l}^d_n$ are blowup coefficcients for $\mathbb{C}^2/\mathbb{Z}_2$ blowup, resulting from $\mathcal{Z}_{1-loop}$.
		
		We convolute relation from these summands with
		\begin{eqnarray}
		\mathcal{Z}(u;q_1q_2,q_1^{-2}|z)=z^{\frac{\log^2u}{8\log q_1 (\log q_2+\log q_1)}}\sum_{p=0}^{+\infty}b_p(u)z^p,\\
		b_0(u)=(q_1 q_2^{-1})^{\frac{\log^2u}{8\log q_1 (\log q_2+\log q_1)}}\mathcal{Z}_{1-loop}(u;q_1q_2,q_1^{-2})
		\end{eqnarray}
		
		(i) In the case $j=0$ the whole sum, as series in $z$, has form 
		\begin{equation}
		z^{-\frac{\log^2u}{8\log q_1 \log q_2}} \sum_{k=0}^{+\infty}c_k(u) \epsilon^k z^{k/2},
		\end{equation}
		so we see that $-2$ bilinear relation splits into relations with integer and half-integer powers of $z$ and sign $\epsilon$
		signifies that.
		
		After convolution we have up to common factor
		\begin{equation}
		\begin{aligned}
		\sum_{m\in\mathbb{Z}}\epsilon^{m} u^{(d_{12})\frac{m}2} (q_1q_2^2)^{(d_{12})\frac{m^2}2} z^{\frac{m^2}2}
		\sum_{p=0}^{+\infty}b_p(uq_1^{2m}q_2^{2m})((q_1q_2)^{(d_{12})}z)^p 
		\sum_{k=0}^{+\infty}c_k(uq_2^{2m}) (q_2^{d_{12}}z)^{k/2}=0 
		\end{aligned}
		\end{equation}
		Then, equating to zero coefficcients in front of power $z^{n/2}$, $n\in\mathbb{Z}_{\geq0}$, due to
		$b_0(u)\neq0$ we obtain $c_0=0, c_1=0, c_2=0 \ldots$ successively.
		
		(ii) In the case $j=1$ there is no such splitting
		\begin{equation}
		z^{\frac18-\frac{\log^2u}{8\log q_1 \log q_2}} \sum_{k=0}^{+\infty}d_k(u) z^{k}.
		\end{equation}
		
		Two vanishing convolutions with $d_{12}\neq d_{12}'$ give us up to common factor
		\begin{eqnarray}
		\sum_{m\in\mathbb{Z}+1/2}\epsilon^{m} u^{d_{12}\frac{m}2} (q_1q_2^2)^{d_{12}\frac{m^2}2} z^{\frac{m^2}2}
		\sum_{p=0}^{+\infty}b_p(u(q_1q_2)^{2m})((q_1q_2)^{d_{12}}z)^p 
		\sum_{k=0}^{+\infty}d_k(uq_2^{2m}) (q_2^{d_{12}}z)^{k}=0 \\
		\sum_{m\in\mathbb{Z}+1/2}\epsilon^{m} u^{d_{12}'\frac{m}2} (q_1q_2^2)^{d_{12}'\frac{m^2}2} z^{\frac{m^2}2}
		\sum_{p=0}^{+\infty}b_p(u(q_1q_2)^{2m})((q_1q_2)^{d_{12}'}z)^p 
		\sum_{k=0}^{+\infty}d_k(uq_2^{2m}) (q_2^{d_{12}'}z)^{k}=0 
		\end{eqnarray}
		Equating to zero coefficcients in front of power $z^{n+1/8}$, $n\in\mathbb{Z}_{\geq0}$, we will obtain $2\times2$ linear system
		on $d_n(uq_2)$ and $d_n(uq_2^{-1})$ with inhomogeneity that equals linear combination of $d_{k}, k<n$.
		Fundamental matrix of this system is
		\begin{equation}
		\begin{pmatrix}
		u^{d_{12}/4} (q_1q_2^2)^{d_{12}/8} b_0(uq_1q_2) q_2^{d_{12}p}& u^{-d_{12}/4} (q_1q_2^2)^{d_{12}/8} b_0(uq_1^{-1}q_2^{-1}) q_2^{d_{12}p}\\
		u^{d_{12}'/4} (q_1q_2^2)^{d_{12}'/8} b_0(uq_1q_2) q_2^{d_{12}'p}& u^{-d_{12}'/4} (q_1q_2^2)^{d_{12}'/8} b_0(uq_1^{-1}q_2^{-1}) q_2^{d_{12'}p}
		\end{pmatrix},
		\end{equation}
		its determinant is equal to
		\begin{equation}
		(q_1q_2^2)^{(d_{12}+d_{12}')/8} b_0(uq_1q_2) b_0(u(q_1q_2)^{-1}) q_2^{(d_{12}+d_{12}')p}(u^{(d_{12}-d_{12}')/4}-u^{(d_{12}'-d_{12})/4}), 
		\end{equation}
		which is nonzero for general $u$. So we obtain $d_0=0, d_1=0, d_2=0 \ldots$ successively.
	\end{proof}

	We see that in case $j=0$ sign $\epsilon$ corresponds to the branch of square root $z^{1/2}$, and we will omit it in next Subsection.

	\subsection{Relations in even sector}
	\label{ssec:even}
	Let us itemize different $-2$ bilinear relations that we could obtain from above approach.
	As explained above, we can take $-2$ bilinear relation consisting only from terms with same $d_1+d_2$
	We know Nakajima-Yoshioka blowup relations \eqref{eq:Z=ZZ} only for cases $d=-1,0,1$, that's our another restriction.
	
	\begin{itemize}
		\item
		\textbf{Case $d_1+d_2=2$.}
		We have only one term with $d_1=d_2=1$ and coefficient in \eqref{coeff} is nonzero, so there is no relation.
		
		\item \textbf{Case $d_1+d_2=1$.}
		We have two terms: with $d_1=1,d_2=0$ and $d_1=0,d_2=1$.
		Both coefficients in \eqref{coeff} equal $1$. We obtain relation
		\begin{equation}
		\begin{aligned}
		\sum_{n\in\mathbb{Z}}\mathcal{Z}(uq_1^{2n};q_1^2,q_2q_1^{-1}|q_1z)\mathcal{Z}(uq_1^{2n};q_1q_2^{-1},q_2^2|q_2z)=
		\sum_{n\in\mathbb{Z}}\mathcal{Z}(uq_1^{2n};q_1^2,q_2q_1^{-1}|q_1^{-1}z)\mathcal{Z}(uq_1^{2n};q_1q_2^{-1},q_2^2|q_2^{-1}z),\label{bilconfrel424}
		\end{aligned}
		\end{equation}
		this is conjectured relation (4.24) from \cite{BGM17}.
		
		\item  \textbf{Case $d_1+d_2=0$.}
		We have three terms: with $d_1=0,d_2=0$, $d_1=1, d_2=-1$ and  $d_1=-1,d_2=1$.
		Coefficcients in \eqref{coeff} are equal to $1$, $1-(q_1q_2)^{-1/2}z^{1/2}$ and $1-(q_1q_2)^{1/2}z^{1/2}$ respectively. We obtain relations 
		\begin{equation}
		\begin{aligned}
		\sum_{n\in\mathbb{Z}}\mathcal{Z}(uq_1^{2n};q_1^2,q_2q_1^{-1}|q_1^{2}z)\mathcal{Z}(uq_1^{2n};q_1q_2^{-1},q_2^2|q_2^{2}z)\\
		=(1-(q_1q_2)^{1/2}z^{1/2})\sum_{n\in\mathbb{Z}}\mathcal{Z}(uq_1^{2n};q_1^2,q_2q_1^{-1}|z)\mathcal{Z}(uq_1^{2n};q_1q_2^{-1},q_2^2|z)\label{bilconfrel}
		\end{aligned}
		\end{equation}
		and
		\begin{equation}
		\begin{aligned}
		\sum_{n\in\mathbb{Z}}\mathcal{Z}(uq_1^{2n};q_1^2,q_2q_1^{-1}|q_1^{-2}z)\mathcal{Z}(uq_1^{2n};q_1q_2^{-1},q_2^2|q_2^{-2}z)\\
		=(1-(q_1q_2)^{-1/2}z^{1/2})\sum_{n\in\mathbb{Z}}\mathcal{Z}(uq_1^{2n};q_1^2,q_2q_1^{-1}|z)\mathcal{Z}(uq_1^{2n};q_1q_2^{-1},q_2^2|z)
		\end{aligned}
		\end{equation}
		which is equivalent to previous relation \eqref{bilconfrel} under substitution $q_1,q_2\mapsto q_1^{-1},q_2^{-1}$
		according to Remark \ref{rem:binv}.
		Relation \eqref{bilconfrel} is conjectured relation (B.5) from \cite{BS16q} and (4.25) from \cite{BGM17}
	\end{itemize}
	
	One can see easily that cases $d_1+d_2=-1,-2$ don't give any new relations according to Remark \ref{rem:binv}.
	
	\subsection{Relations in odd sector}
	\label{ssec:odd}
	According to Prop. \ref{prop:suff} (ii) we start from itemizing different vanishing convolutions. 
	After that we will look for appropriate pairs of convolutions and write corresponding $-2$ bilinear relations. 
	\begin{itemize}
		\item \textbf{Case $d_1+d_2=2$.}
		We have only one term with $d_1=d_2=1$ and coefficient in \eqref{coeff} is $-(q_2^{3}q_1z)^{1/4}(1+\epsilon)$,
		so we have vanishing convolution only for $\epsilon=-1$.
		\item \textbf{Case $d_1+d_2=1$.}
		We have two terms: with $d_1=1,d_2=0$ and $d_1=0,d_2=1$.
		Coefficients in \eqref{coeff} equal $-(q_2^2z)^{1/4}$ and $-\epsilon (q_1q_2)^{1/4} (q_2^2z)^{1/4}$ respectively,
		so we have two vanishing convolutions with $\epsilon=\pm1$.
		
		\item \textbf{Case $d_1+d_2=0$.}
		We have three terms: with $d_1=0,d_2=0$, $d_1=1,d_2=-1$ and $d_1=-1,d_2=1$.
		Coefficcients are equal to $0$, $(q_2q_1^{-1} z)^{1/4}(\epsilon q_2^{-1}-1)$ and $\epsilon q_2(q_2q_1^{-1} z)^{1/4}(\epsilon q_2^{-1}-1)$
		respectively, so we have vanishing convolutions consisting from first term or from next two terms, both with $\epsilon=\pm1$. 
		
		\item \textbf{Case $d_1+d_2=-1$.}
		We have two terms: with $d_1=-1,d_2=0$ and $d_1=0,d_2=-1$.
		Coefficients in \eqref{coeff} equal $(q_2^{-2}z)^{1/4}$ and $\epsilon (q_1q_2)^{-1/4} (q_2^{-2}z)^{1/4}$ respectively,
		so we have two vanishing convolutions with $\epsilon=\pm1$.
		
		\item \textbf{Case $d_1+d_2=-2$.}
		We have only one term with $d_1=d_2=-1$ and coefficient in \eqref{coeff} is $(q_2^{-3}q_1^{-1}z)^{1/4}(1+\epsilon)$, so we have vanishing
		convolution only for $\epsilon=-1$.
	\end{itemize}
	
	We see that we have $-2$ bilinear relations, which follow from vanishing conditions in the cases $d_1+d_2=\pm1$ both for $\epsilon=\pm1$
	\begin{equation}
	\begin{aligned}\label{bilconfrel423}
	\sum_{n\in\mathbb{Z}+1/2}\epsilon^n\mathcal{Z}(uq_1^{2n};q_1^2,q_2q_1^{-1}|q_1z)\mathcal{Z}(uq_1^{2n};q_1q_2^{-1},q_2^2|q_2z)\\=
	\epsilon (q_1q_2)^{1/4} \sum_{n\in\mathbb{Z}+1/2}\epsilon^n\mathcal{Z}(uq_1^{2n};q_1^2,q_2q_1^{-1}|q_1^{-1}z)\mathcal{Z}(uq_1^{2n};q_1q_2^{-1},q_2^2|q_2^{-1}z).
	\end{aligned}
	\end{equation}
	This is just relation (4.23) from \cite{BGM17}, it could be seen by taking sum and difference of relations with two different $\epsilon$. 
	
	We also have $-2$ bilinear relations, which follow from vanishing conditions in cases $d_1+d_2=0$ and $d_1+d_2=\pm2$ only for $\epsilon=-1$ 
	\begin{equation}
	\sum_{n\in\mathbb{Z}+1/2}(-1)^n\mathcal{Z}(uq_1^{2n};q_1^2,q_2q_1^{-1}|z)\mathcal{Z}(uq_1^{2n};q_1q_2^{-1},q_2^2|z)=0. \label{bilconfrel422}
	\end{equation}
	This is just relation (4.22) from \cite{BGM17}.
	
	This exhausts the relations, which follow from the above vanishing convolutions.
	
	\subsection{Chern-Simons modification}
	\label{ssec:CS}
	Obtaining $-2$ bilinear relations on $\mathcal{Z}^{[l]}(u;q_1,q_2|z)$ by our approach seems to be much more subtle and cumbersome. 
	Let us illustrate this approach on relation 
	\begin{equation}
	\begin{aligned}
	\sum_{n\in\mathbb{Z}}\mathcal{Z}^{[1]}(uq_1^{2n};q_1^2,q_2q_1^{-1}|q_1^{2}z)\mathcal{Z}^{[1]}(uq_1^{2n};q_1q_2^{-1},q_2^2|q_2^{2}z)\\
	=(1-(q_1q_2)^{1/2}z^{1/2})\sum_{n\in\mathbb{Z}}\mathcal{Z}^{[1]}(uq_1^{2n};q_1^2,q_2q_1^{-1}|q_1z)\mathcal{Z}^{[1]}(uq_1^{2n};q_1q_2^{-1},q_2^2|q_2z),
	\label{bilconfrel1}
	\end{aligned}
	\end{equation}
	which for $q_1q_2=1$ was proposed in \cite{BGM18} (in terms of tau functions for $k=1, N=2$ see formula (3.7) in loc. cit.).
	
	We should modify our convolution \eqref{convol} for $l=1$ in the following way. Split our $-2$ bilinear relation into two parts with integer
	and half-integer powers of $z$ (up to common factor)
	\begin{equation}
	z^{-\frac{\log^2u}{8\log q_1 \log q_2}} (\sum_{k=0}^{+\infty}c_{2k}(u) z^{k}+\epsilon\sum_{k=0}^{+\infty}c_{2k+1}(u)  z^{k+1/2})=
	\widehat{Z}^{[1]}_{0}(z)+\widehat{Z}^{[1]}_{1}(z)
	\end{equation}
	Then make a modified convolution (braces denote fractional part of the number)
	\begin{equation}
	\sum_{m\in\mathbb{Z}, p=0,1}  \mathcal{Z}^{[1]}(u(q_1q_2)^{2m};q_1q_2,q_1^{-2}|(q_1q_2)^{l(2\{(m+p)/2\}-1)}z)\widehat{Z}^{[1]}_{n}(uq_2^{2m},q_1,q_2;q_2^{l(2\{(m+p)/2\}-1)}z),\label{CSconvol}
	\end{equation}
	namely, convolution shift become dependent on relative parity of $m$ and $p$
	
	As before, using Nakajima-Yoshioka blowup relations \eqref{eq:Z=ZZ} twice we obtain 
	\begin{equation}
	\beta^{-1}_0 \beta^1_0+\beta^{-1}_1 (q_1q_2,q_2 q_1^{-1}|q_1 q_2 z) \beta_1^1(q_1q_2,q_2^2|z)=\beta_0^0 \beta_0^0+\beta_1^0 \beta_1^0-(q_1q_2)^{1/2}z^{1/2}(\beta_0^0 \beta_0^1+q_2^{-1/2}\beta_1^{-1} \beta_1^0), 
	\end{equation}
	where we wrote dependence on variables only where it is necessary. This expression turns out to be an identity.
	
	Finally, above convolution, written as sum of $z$-powers, is
	\begin{equation}
	\begin{aligned}
	\sum_{m\in\mathbb{Z}}\epsilon^{m}  z^{\frac{m^2}2}
	\sum_{p=0}^{+\infty}b_p(uq_1^{2m}q_2^{2m})((q_1q_2)^{l(2\{(m+p)/2\}-1)}z)^p 
	\sum_{k=0}^{+\infty}c_k(uq_2^{2m}) (q_2^{l(2\{(m+p)/2\}-1)}z)^{k/2}=0, 
	\end{aligned}
	\end{equation}
	and, as before, $c_n=0, n\geq0$ successively as before. So we obtained the proof of \eqref{bilconfrel1}.
	
	\section{Applications}
	\label{sec:appl}
	\subsection{Nakajima-Yoshioka blowup relations of higher order}
	\label{ssec:hoNY}
	Using our approach in opposite direction, we could obtain Nakajima-Yoshioka blowup relation \eqref{eq:Z=ZZ}
	for $l=0$, $j=0$, $d=2$. We could always write such relation with unknown function $\beta_2^{0,1}$, which apriori
	is a power series in $z$ and dependens on $u$. Let's calculate it. 
	Let us take convolution \eqref{convol} of \eqref{bilconfrel}, taking $d_1=0, d_2=2$ and $d_1=d_2=1$ for the corresponding summands of the relation.
	Because \eqref{bilconfrel} is already proved, reducing convolution to linear combination of Nekrasov partition functions as was done
	in Subsection \ref{ssec:gs}, we obtain
	\begin{equation}
	\beta^0_0 \beta^2_0(u;q_1q_2,q_2^2|z)+\beta^0_1 \beta^2_1=(1-(q_1q_2^3z)^{1/2}) (\beta^1_0\beta^1_0+\beta^1_1(q_1q_2, q_2q_1^{-1}|q_1q_2z)\beta^1_1
	(q_1q_2, q_2^2|z))=1-q_1q_2^3 z, 
	\end{equation}
	where we wrote dependence on variables only where it is necessary. So, thanks to $\beta_1^0=0$, we find $\beta^2_0$
	\begin{equation}
	\beta^2_0(q_1,q_2|z)=1-q_1q_2z \label{beta20}
	\end{equation}
	and, according to \eqref{betasymm}
	\begin{equation}
	\beta^{-2}_0(q_1,q_2|z)=1-q_1^{-1}q_2^{-1}z. \label{betam20}
	\end{equation}
	
	In the same manner we can find $\beta^{2,[1]}_0(q_1,q_2|z)$ for $l=1$ (in fact, it differs from above $\beta^2_0$). Analogous
	calculation with \eqref{bilconfrel1} and convolution \eqref{CSconvol}, where other shift $l(2\{(m+p)/2\}-1)\mapsto 2+l(2\{(m+p)/2\}-1)$ is taken,
	we obtain
	\begin{equation}
	\beta^{0}_0 \beta^2_0(u;q_1q_2, q_2^2|z)+\beta^0_1 \beta^2_1=
	\beta_0^1 \beta_0^1+\beta_1^1(q_1q_2, q_2q_1^{-1}|q_1q_2 z) \beta_1^1 (q_1q_2, q_2^2|z)-(q_1q_2^3)^{1/2}z^{1/2}
	(\beta_0^1 \beta_0^2(u;q_1q_2, q_2^2|z)+q_2^{-1}\beta_1^0 \beta_1^1), 
	\end{equation}
	so we find that $\beta^{2,[1]}_0(q_1,q_2|z)=1$, which differ from \eqref{beta20}.
	
	\subsection{Proof of equalities \eqref{Z2=Z0} and \eqref{qinv1}}
	\label{ssec:bruteforce}
	
	This Subsection starts from proving equality \eqref{Z2=Z0}, using Nakajima-Yoshioka
	blowup equations \eqref{eq:Z=ZZ} for $l=2$ and $l=0$ in sector $j=0$.
	
	As already mentioned in \cite{BS18}, in terms of topological strings this relation means a relation between
	the geometry of local $\mathbb{F}_0=\mathbb{P}^1\times\mathbb{P}^1$ and local Hirzebruch surface $\mathbb{F}_2$.
	Particularly, the relation between Gopakumar-Vafa invariants of these manifolds is given in e.g. \cite[eq. (94)]{IKP02}.
	Case $q_1q_2=1$ of \eqref{Z2=Z0} appear in \cite{BGM18}, it is proved together with cases $q_1^2q_2=1, q_1q_2^2=1$
	in \cite[Prop. 4.3.]{BS18}. Before this papers we have not found equality \eqref{Z2=Z0} in the literature, but it is maybe known.
	For full Nekrasov functions we have the same equality, because $\mathcal{Z}_{cl}$ and $\mathcal{Z}_{1-loop}$
	do not depend on $l$
	\begin{prop}\label{20equivprop}
		Nekrasov function $\mathcal{Z}^{[2]}$ is equal to $\mathcal{Z}^{[0]}$ up to double $q$-Pochhammer symbol
		\begin{equation}
		\mathcal{Z}^{[2]}(u;q_1,q_2|z)=(z;q_1,q_2)_{\infty}\mathcal{Z}^{[0]}(u;q_1,q_2|z), \label{Z2=Z0:fullz}
		\end{equation}
	\end{prop}
	We follow proof of \cite[Prop. 4.3.]{BS18}, sligthly modifying it for arbitrary $q_1,q_2$.
	
	Consider Nakajima-Yoshioka blowup relations \eqref{eq:Z=ZZ} for $l=2$, $j=0$, $d=-1,0,1$ and exclude 
	$\mathcal{Z}^{[2]}(u;q_1,q_2|z)$ from them
	\begin{equation}
	\begin{aligned}
	&\sum_{n\in\mathbb{Z}}\mathcal{Z}^{[2]}(uq_1^{2n};q_1,q_2q_1^{-1}|q_1^{-2}z) \mathcal{Z}^{[2]}(uq_2^{2n};q_1q_2^{-1},q_2|q_2^{-2}z)=\\
	=&\sum_{n\in\mathbb{Z}}\mathcal{Z}^{[2]}(uq_1^{2n};q_1,q_2q_1^{-1}|q_1^{-1}z) \mathcal{Z}^{[2]}(uq_2^{2n};q_1q_2^{-1},q_2|q_2^{-1}z)=\\
	=&\sum_{n\in\mathbb{Z}}\mathcal{Z}^{[2]}(uq_1^{2n};q_1,q_2q_1^{-1}|z) \mathcal{Z}^{[2]}(uq_2^{2n};q_1q_2^{-1},q_2|z)\label{qCSdetermrel}.
	\end{aligned}
	\end{equation}
	
	These equations are equations on $\mathcal{Z}^{[2]}_{inst}$, or bilinear equations
	on coefficients $c_k^{(1)}$, $c_k^{(2)}, k\in\mathbb{Z}_{\geq 0}$
	of the corresponding power series 
	$\mathcal{Z}^{[2]}_{inst}(u;q_1,q_2q_1^{-1}|z)=\sum_{k=0}^{+\infty} c_k^{(1)}(u;q_1,q_2) z^{k}$ and 
	$\mathcal{Z}^{[2]}_{inst}(u;q_1q_2^{-1},q_2|z)=\sum_{k=0}^{+\infty} c_k^{(2)}(u;q_1,q_2) z^{k}$.
	Namely, relations \eqref{qCSdetermrel} split into the relations corresponding to powers $z^{k}, k\in\mathbb{Z}_{\geq0}$ 
	(up to the power $\Lambda^{-\frac{\log^2u}{\log q_1\log q_2}}$ from $\mathcal{Z}_{cl}$). 
	
	\begin{lemma}(\cite[Lemma 4.1.]{BS18})
		Relations \eqref{qCSdetermrel} recursively determine coefficients $c_k^{(1)}, c_k^{(2)}, k\in\mathbb{N}$
		starting from normalization  $c_0^{(1)}=c_0^{(2)}=1$.\label{determlemma}
	\end{lemma}
	
	\begin{proof}
		Let us take the coefficient of the power $z^{k}$ in the relation \eqref{qCSdetermrel}, then coefficients $c_k^{(1)}$, $c_k^{(2)}$ 
		appear only for $n=0$. Other coefficients in these relations have lower index, so they are known due to the induction supposition.
		Therefore we obtain system of two linear equations on two unknown variables $c^{(1)}_k$, $c_k^{(2)}$.
		The fundamental matrix of this system is as follows
		\begin{equation}
		\begin{pmatrix}
		q_1^{-k}-1 & q_2^{-k}-1\\
		q_1^{-2k}-1 & q_2^{-2k}-1 \label{fm2}
		\end{pmatrix}, 
		\end{equation}
		its determinant equals $(q_1^{-k}-1)(q_2^{-k}-1)(q_2^{-k}-q_1^{-k})$
		which is non-zero iff none of $q_1,q_2, q_1/q_2$ is a root of unity. 
	\end{proof}
	In our sector $|q_1|\lessgtr1, |q_2|\gtrless1$, these special cases are not realized.
	
	\begin{proof}[Proof of the Proposition \ref{20equivprop}]
		Consider Nakajima-Yoshioka blowup relations \eqref{eq:Z=ZZ} for $l=0$, $j=0$, $d=-2,-1,0$ (where we
		needed \eqref{betam20}, found above by our approach) and exclude from
		them $\mathcal{Z}^{[2]}(u;q_1,q_2|z)$
		\begin{equation}
		\begin{aligned}
		&(1-q_1^{-1}q_2^{-1}z)^{-1}\sum_{n\in\mathbb{Z}}\mathcal{Z}^{[0]}(uq_1^{2n};q_1,q_2q_1^{-1}|q_1^{-2}z) \mathcal{Z}^{[0]}(uq_2^{2n};q_1q_2^{-1},q_2|q_2^{-2}z)=\\
		=&\sum_{n\in\mathbb{Z}}\mathcal{Z}^{[0]}(uq_1^{2n};q_1,q_2q_1^{-1}|q_1^{-1}z) \mathcal{Z}^{[0]}(uq_2^{2n};q_1q_2^{-1},q_2|q_2^{-1}z)=\\
		=&\sum_{n\in\mathbb{Z}}\mathcal{Z}^{[0]}(uq_1^{2n};q_1,q_2q_1^{-1}|z) \mathcal{Z}^{[0]}(uq_2^{2n};q_1q_2^{-1},q_2|z).
		\end{aligned}
		\end{equation}
		
		Let us replace $\mathcal{Z}^{[0]}$ with $\td{\mathcal{Z}}^{[2]}$, formally defined by \eqref{Z2=Z0:fullz}. 
		From calculation with $q$-Pochhammers
		\begin{equation}
		\frac{(q_1^{-1}z;q_1,q_2q_1^{-1})_{\infty}(q_2^{-1}z;q_1q_2^{-1},q_2)_{\infty}}
		{(z;q_1,q_2q_1^{-1})_{\infty}(z;q_1q_2^{-1},q_2)_{\infty}}=1, \quad
		\frac{(q_1^{-2}z;q_1,q_2q_1^{-1})_{\infty}(q_2^{-2}z;q_1q_2^{-1},q_2)_{\infty}}
		{(z;q_1,q_2q_1^{-1})_{\infty}(z;q_1q_2^{-1},q_2)_{\infty}}=\frac{1}{1-q_1^{-1}q_2^{-1}z},
		\end{equation}
		we obtain that $\td{\mathcal{Z}}^{[2]}$ satisfies \eqref{qCSdetermrel}.
		Therefore, according to Lemma \ref{determlemma} $\td{\mathcal{Z}}^{[2]}=\mathcal{Z}^{[2]}$  (for general $q_1,q_2$),
		which is desired relation \eqref{Z2=Z0:fullz}.
	\end{proof}
	
	Relation \eqref{qinv1} is obtained, using the same idea, but even simpler.
	\begin{prop}
		Nekrasov instanton partition function  $\mathcal{Z}^{[1]}_{inst}$ is invariant under $q_1,q_2\mapsto q_1^{-1},q_2^{-1}$
		\begin{equation}
		\mathcal{Z}^{[1]}_{inst}(u;q_1,q_2|z)=\mathcal{Z}^{[1]}_{inst}(u;q_1^{-1},q_2^{-1}|z).
		\end{equation}
	\end{prop}
	
	\begin{proof}
		Using \eqref{CSlinv}, we find out that we should prove
		\begin{equation}
		\mathcal{Z}^{[-1]}_{inst}(u;q_1,q_2|z)=\mathcal{Z}^{[1]}_{inst}(u;q_1,q_2|z)
		\end{equation}
		For $l=1$ take Nakajima-Yoshioka blowup relations $d=0,1,2$ and for $l=-1$ take relations $d=-1,0,1$.
		Excluding $\mathcal{Z}^{[\pm1]}(u;q_1,q_2|z)$ from these relations we obtain two bilinear relations on 
		$\mathcal{Z}^{[l]}(u;q_1,q_2q_1^{-1}|z)$ and $\mathcal{Z}^{[\l]}(u;q_1q_2^{-1},q_2|z)$ both for $l=-1$ and $l=1$. 
		And these pairs of relations are identical. Here analogue of fundamental matrix \eqref{fm2} is matrix
		\begin{equation}
		\begin{pmatrix}
		q_1^{k/2}-q_1^{-k/2} & q_2^{k/2}-q_2^{-k/2}\\
		q_1^{3k/2}-q_1^{-k/2} & q_2^{3k/2}-q_2^{-k/2} \label{fm1}
		\end{pmatrix}, 
		\end{equation}
		which determinant is equal to $(q_1q_2)^{-k/2}(q_1^{-k}-1)(q_2^{-k}-1)(q_2^{-k}-q_1^{-k})$ which is non-zero iff none
		of $q_1,q_2, q_1/q_2$ is a root of unity. This completes the proof.
	\end{proof}
	\section{Further questions}
	\begin{itemize}
		\item 
		Relation \eqref{bilconfrel} was already written in terms of quantum tau functions, i.e., tau
		functions of form \eqref{Mastertau}, where $[\sg,\log s]=\hbar$ (see (4.15) in \cite{BGM17})
		It will be interesting to rewrite our approach in terms of certain quantum tau functions.
		\item Using this approach, we possibly could find $-3$, $-4$ \ldots bilinear relations on Nekrasov partition functions,
		which result from blowup relations on orbifolds $\mathbb{C}^2/\mathbb{Z}_p, \, p>2$.
		\item In terms of quantum tau functions this approach probably will be easy to generalize for
		$SU(N)$, $N>2$ Nekrasov partition functions. Other possibly useful generalizations
		include adding matter supermultiplets, circular quiver gauge theories etc.
	\end{itemize}

	\noindent\textsc{National Research University Higher School of Economics, Moscow, Russia\\
		Center for Advanced Studies, Skolkovo Institute of Science and Technology, Moscow, Russia}
	
	\emph{E-mail}:\,\,\textbf{shch145@gmail.com}

\end{document}